\documentclass[12pt]{article}
\usepackage{fullpage}

\usepackage{amsmath}
\usepackage{amssymb}
\usepackage{amsthm}
\usepackage{mathrsfs}
\usepackage{prooftree}
\usepackage{multirow}
\usepackage{hyperref}
\usepackage{subcaption}
\usepackage{tikz}
\usepackage{tabularx}
\usepackage{pgfplots}
\usetikzlibrary{arrows,shapes.multipart, matrix}

\newtheorem{definition}{{Definition}}

\newtheorem{theorem}[definition]{{Theorem}}

\newcommand{\ag}{\textit{Ag}}
\newcommand{\nonces}{\mathscr{N}}
\newcommand{\keys}{\mathscr{K}}

\newcommand{\priv}{\textit{sk}}

\newcommand{\inv}{\textit{inv}}
\newcommand{\basic}{\mathscr{B}}
\newcommand{\asserts}{\mathscr{A}}

\newcommand{\disj}{\lor}
\newcommand{\conj}{\land}

\newcommand{\says}{\textit{says}}
\newcommand{\pipe}{\ \mid\ }

\newcommand{\rnax}{\textit{ax}}

\newcommand{\rncontr}{\bot}

\newcommand{\rnandintro}{\wedge{i}}
\newcommand{\rnandelim}{\wedge{e}}
\newcommand{\rnorintro}{\vee{i}}
\newcommand{\rnorelim}{\vee{e}}

\newcommand{\rnexistsintro}{\exists{i}}
\newcommand{\rnexistselim}{\exists{e}}

\newcommand{\rnpair}{\textit{pair}}
\newcommand{\rnsplit}{\textit{split}}
\newcommand{\rnenc}{\textit{enc}}
\newcommand{\rndec}{\textit{dec}}

\newcommand{\rnsays}[1]{\textit{says}_{#1}}

\newcommand{\derives}{\vdash}
\newcommand{\DY}{\textit{dy}}
\newcommand{\DYderives}{\vdash_{\DY}}

\newcommand{\prot}{\textit{Pr}}

\newcommand{\ks}{\textit{ks}}

\newcommand{\setofterms}{\mathscr{T}}

\newcommand{\zkp}{\text{ZK}}
\newcommand{\blind}{\text{blind}}

\newcommand{\commit}{\text{commit}}

\newcommand{\sent}{\textit{sent}}

\newcommand{\vars}{\mathscr{V}}

\newcommand{\sendact}[4]{+{#1}\!\!:({#2})({#3},{#4})}
\newcommand{\recact}[3]{-{#1}\!\!:({#2},{#3})}
\newcommand{\confirmact}[2]{{#1}\!\!:\textit{confirm}\ {#2}}
\newcommand{\denyact}[2]{{#1}\!\!:\textit{deny}\ {#2}}
\newcommand{\insertact}[2]{{#1}\!\!:\textit{insert}\ {#2}}
\newcommand{\agentid}{\textit{id}}

\newcommand{\hidden}[1]{#1}

\newcommand{\swp}{\textsf{swp}}
\usepackage{cellspace}
\begin{document}

\title{Existential Assertions for Voting Protocols}
\author{
R. Ramanujam\\
Institute of Mathematical Sciences\\
Chennai, India. \\
\texttt{jam@imsc.res.in}
\and 
Vaishnavi Sundararajan\\
Chennai Mathematical Institute,
Chennai, India. \\
\texttt{vaishnavi@cmi.ac.in}
\and S.P. Suresh \\
Chennai Mathematical Institute,
Chennai, India. \\
\texttt{spsuresh@cmi.ac.in}
}

\date{}

\maketitle

\pagestyle{plain}
\thispagestyle{empty}

 \begin{abstract}
In~\cite{RSS14}, we extend the Dolev-Yao model with assertions. We build on that work and add existential abstraction to the language, which allows us to translate common constructs used in voting protocols into proof properties. We also give an equivalence-based definition of anonymity in this model, and prove anonymity for the FOO voting protocol.
\end{abstract}

\section{Anonymity}
Formal verification of security protocols often involves the analysis of a property where the relationship between an agent and a message sent by him/her needs to be kept secret. This property, called ``anonymity'', is  a version of the general unlinkability property, and one of much interest. There can be multiple examples of such anonymity requirements, including healthcare records, online shopping history, and movie ratings~\cite{NP09}. Electronic voting protocols are a prime example of a field where ensuring and verifying anonymity is crucial. 

It is interesting to see how protocols are modelled symbolically for the analysis of such properties. In the Dolev-Yao model~\cite{DY83}, one often requires special operators in order to capture certain behaviour. Many voting schemes employ an operation known as a {\em blind signature}~\cite{Ch83}. A blind signature is one where the underlying object can be hidden (via a blinding factor), the now-hidden object signed, and then the blind removed to have the signature percolate down to the underlying object. The FOO voting protocol given in~\cite{FOO92} crucially uses blind signatures in order to obtain a signature on an encrypted object. \cite{BRS10} shows that the derivability problem for protocols involving blind signatures becomes DEXPTIME-hard. Protocols which do not use blind signatures often use homomorphic encryption or mix nets, which also make the modelling and verification quite complex~\cite{LLT07}.

Note that in most common models, terms are the only objects communicated. A ``certificate'' of an agent's validity -- which is an intrinsically different object from a term containing an agent's vote, for example -- is also modelled as a term in the term algebra. \cite{BHM08,BMU08}, for example, augment the Dolev-Yao term syntax with an extra primitive $\zkp$, which can be used to create a term that codes up a zero-knowledge proof. However, no direct logical inference is possible with these proof terms, and therefore, it is difficult to reason about what further knowledge agents can obtain using them. \cite{RSS14} proposes a departure from this paradigm, using {\em assertions} as a further abstraction that can be used for modelling protocols. Assertions, which code up certificates and have a separate proof system, can be sent by agents in addition to terms. The assertion algebra allows designers to model protocols involving certification in a more explicatory manner (by maintaining terms and certificates as separate objects). It also allows analysts to capture any increase in agents' knowledge achieved by deduction at the level of certificates. 

So what are these assertions and how do they behave? Assertions include statements about various terms appearing in the protocol. These include instances of application-specific predicates and equalities between two different symbolic terms. Assertions can also be combined using the usual propositional connectives \textit{and} ($\conj$) and \textit{or} ($\disj$). They also include a $\says$ operator, which works as an ownership mechanism for assertions, and disallows other agents from forwarding such an assertion in their own name. Perhaps the most crucial (and useful) addition to the assertion language here (over the system in~\cite{RSS14}) is the existential quantifier. This allows us to quantify out any term from an assertion, thereby effectively hiding the actual term about which that assertion is made. Since existential assertions thus hide the private data used to generate a certificate, while revealing some partial information, they seem especially useful for capturing blinding (and similar operations with this goal) in voting protocols.   

\subsection{Related Work}
Research on anonymity has been carried out for many years now. In the applied-pi calculus, \cite{KR05} verifies anonymity for the FOO protocol, \cite{ACRR10} studies general unlinkability and shows that this implies anonymity, and \cite{MVdV07} provides an applied-pi based model incorporating aspects of the underlying communication mechanism (anonymous channels in particular). 

There are also many epistemic logic-based approaches. \cite{HoN05} provides a logical framework built on modal epistemic logic for anonymity in multiagent systems; \cite{GS98,SS99} also define information-hiding properties in terms of agent knowledge; \cite{HS04} provides a modular framework that allows one to analyze general unlinkability properties using function views, along with extensive case studies on anonymity and privacy.

Theorem provers have also been used to verify anonymity. \cite{BJL12} uses an automatic theorem prover MCMAS for verification; \cite{BGB13} also specifies general unlinkability as an extension to the Inductive Method for security protocol verification in the theorem prover Isabelle.

In this paper, we extract a logical core of reasoning about certificates, translate the typical constructs used for voting protocols into proof properties, and employ equivalence-based reasoning for verifying anonymity. We also apply this technique to model two voting protocols, namely FOO and Helios, and to analyze anonymity for FOO.

\section{Modelling the FOO protocol}
\label{sec:foo}
\subsection{Introduction to FOO}
In~\cite{FOO92}, the authors introduce the FOO protocol for electronic voting, which has inspired many subsequent protocols. This protocol uses blinding functions and bit commitments in order to satisfy many desirable security properties, including anonymity. The voter $V$ sends to the authority $A$ his name, along with a blindsigned commitment to the vote $v$. The authority signs this term, and sends it back to $V$. $V$ now unblinds this to obtain a signature on his commitment to the vote $v$, and sends that to the collector $C$. $C$ adds the encrypted vote and $V$'s commitment to the public bulletin board. $V$ then sends to $C$ the random bit $r$ he used to create the vote commitment, so $C$ can access the vote and update his tally. The protocol is presented in Figure~\ref{fig:footerms} (see~\cite{FOO92,KR05} for a detailed explanation). Sends marked by $\looparrowright$ are over anonymous channels.

\begin{figure}\small
\begin{center}
\begin{subfigure}{0.35\linewidth}
\begin{eqnarray*}
V \rightarrow A & : & V, [\blind(\commit(v, r), b)]_{V} \\
& & \\
& & \\
& & \\
A \rightarrow V & : & [\blind(\commit(v, r), b)]_{A} \\
& & \\
& & \\
& & \\
V \looparrowright C & : & [\commit(v, r)]_{A} \\
C \rightarrow \hphantom{V}  & : & \text{list}, [\commit(v, r)]_{A} \\
V \rightarrow C & : & r \\
& & \\
& & \\
& &
\end{eqnarray*}
\caption{FOO Protocol with terms. $[x]_{A}$ denotes $x$ signed by $A$.}
\label{fig:footerms} 
\end{subfigure}
\hspace{0.08\linewidth}
\begin{subfigure}{0.5\linewidth}
\begin{eqnarray*}
V \rightarrow A &:& \{v\}_{r_{A}}, V\ \says\ \{\ \exists x, r: \{x\}_{r} = \{v\}_{r_{A}} \\
& & \hspace{0.3\textwidth} \conj\ \texttt{valid}(x)\ \ \} \\
A &: & \textit{deny}\ \ \exists x: \texttt{voted}(V, x) \\
A &: & \textit{insert}\ \texttt{voted}(V, \{v\}_{r_{A}}) \\ 
A \rightarrow V &:& A\ \says \\
& & [\ \texttt{elg}(V) \conj \texttt{voted}(V, \{v\}_{r_{A}}) \\
                & & \hspace{0.02\textwidth} \conj\ V\ \says\ \{\ \exists x, r: \{x\}_{r} = \{v\}_{r_{A}} \\
                & & \hspace{0.3\textwidth} \conj\ \texttt{valid}(x)\ \}\ ] \\
V \looparrowright C &:& \{v\}_{r_{C}}, r_{C}, \\
& & \exists X\ \exists y,s: A\ \says \\
& & [\ \texttt{elg}(X) \conj\  \texttt{voted}(X, \{y\}_{s}) \\
& & \hspace{0.02\textwidth} \conj\ X\ \says\ \{\ \exists x, r:\{x\}_{r} = \{y\}_{s} \\
& & \hspace{0.3\textwidth} \conj\ \texttt{valid}(x)\ \}\ ] \\
& & \hspace{0.01\textwidth} \conj\ y = v \\
\end{eqnarray*}
\caption{FOO Protocol with assertions.\\}
\label{fig:fooassertions}
\end{subfigure}
\end{center}
\caption{FOO Protocol: Modelling with terms only and with assertions}
\label{fig:foo}
\end{figure}

\subsection{Modelling FOO with Assertions}
\label{subsec:fooasserts}
In Figure~\ref{fig:fooassertions}, we present the FOO protocol as modelled using assertions.

The voter $V$ contacts the authority $A$ with his vote $v$ encrypted using a random key $r_{A}$. $V$ also sends a certificate linking his name to his encrypted vote $v$. The $V\ \says$ prefix links $V$ to the certificate about $v$, and thus informs the authority that $V$ wishes to vote using the valid vote $v$, the encrypted form of which has been sent with the certificate. Note that this certificate automatically rules out replay attacks (of the kind where another agent $V'$ copies $V$'s published data off the bulletin board and replays it in her own name).

The authority $A$ checks that the voter $V$ has not voted earlier. If this check passes, $A$ adds the fact that $V$ has voted with the encrypted term $\{v\}_{r_{A}}$ to her database (so that $V$ cannot vote again in the future) via an \textit{insert} action. $A$ then issues a certificate stating that $V$ is a valid voter and wishes to vote with the encrypted term he sent $A$ earlier, and that $V$ claims that the term encrypted therein is a valid vote. The voter $V$ now anonymously sends to the counter $C$ the vote $v$ encrypted in a new random key. This is accompanied by an existential assertion, which hides the voter's identity from $C$, while still convincing $C$ that $A$ has certified $V$ and the sent vote to be valid.

We need three predicates here -- \texttt{valid}, \texttt{elg},  and \texttt{voted}. The first two are predicates for stating the validity of the vote and the eligibility of the voter, respectively. The \texttt{voted} predicate is used for linking the voter and the vote. As we shall see in Section~\ref{sec:assertheory}, we can add such protocol-specific predicates to the assertion language in order to communicate succinct certificates (for example, here we use $\texttt{valid}(v)$, instead of providing a disjunction over the finite set of valid votes for the value of $v$, which would grow longer as the set of allowable values grows larger).

\section{Modelling Helios 2.0}
\subsection{Introduction to Helios}
\cite{Ad08} introduces the voting scheme called Helios which has the desirable property of public auditability, i.e., even if Helios is fully corrupt, one can verify the integrity of an election outsourced to it. Helios provides unconditional integrity, while privacy is guaranteed if one trusts the Helios server, which doubles up as election administrator and trustee. The voter sends his vote to the Ballot Preparation System, which creates an encrypted ballot, which is then sealed and cast. The voter's identity and ballot are then posted on the public bulletin board. On closing the election, Helios removes voter names, shuffles all ballots, produces a proof of correct shuffling, and posts these on the board. After allowing some time for auditors to check the shuffling, Helios decrypts each  ballot, produces a proof of correct decryption, and posts the tally on the bulletin board. Helios crucially uses auditing by various participants in order to guarantee correctness.

\subsection{Helios 2.0}
\cite{CS11} demonstrates an attack on vote privacy in the basic Helios system in~\cite{Ad08}, where, by controlling more than half the voters, an adversary can get the compromised voters to copy a single (honest) voter's encrypted ballot off the bulletin board, and from the tally know whom that voter voted for. Note that this happens in spite of the Helios system itself being non-corrupt. In order to fix this, they introduce measures to weed out replayed ballots, and a linking mechanism between every ballot and the voter whose vote it is supposed to encrypt. They also replace the shuffling mechanism by a homomorphic encryption operation, and introduce trustees who are distinct from the election administrator. This introduces an extra assurance of vote privacy, since a corrupt administrator needs to corrupt 
some trustees in order to see a voter's unencrypted vote.

\begin{figure}[h]
\begin{center}
\begin{eqnarray*}
V \rightarrow S &:& v,\ V\ \says\ \texttt{valid}(v) \\
S \rightarrow V &:& b,\ S\ \says\ \{\exists v: b = \textit{ballot}(v) \conj V\ \says\ \texttt{valid}(v) \} \\
V \rightarrow S &:& \textit{cast} \\
S \rightarrow A &:& b,\ S\ \says\ \{\exists v: b = \textit{ballot}(v) \conj V\ \says\ \texttt{valid}(v) \} \\
A &:& \textit{deny}\ {\texttt{voted}(V)} \\
A &:& \textit{insert}\ {\texttt{voted}(V)} \\
A \rightarrow BB &:& b,\ A\ \says\ S\ \says\ \{\exists v: b = \textit{ballot}(v) \conj V\ \says\ \texttt{valid}(v) \} \\
\text{Suppose} & b_{1},\ldots, b_{k} & \text{were the ballots cast and published on the bulletin board.} \\
A \rightarrow BB &:& t,\ A\ \says\ [\exists s: t = \textit{ballot}(s) \conj \\
& & \{ \exists v_{1},\ldots, v_{k}: s = \textit{sum}(v_{1},\ldots, v_{k}) \conj \bigwedge_{i=1}^{k} b_{i} = \textit{ballot}(v_{i}) \} ]  
\end{eqnarray*}
\end{center}
\caption{Helios 2.0 Protocol with assertions}
\label{fig:heliosassertions}
\end{figure}

\subsection{Modelling Helios 2.0 with Assertions}
The voter first inputs his vote to a script which creates his ballot and sends it back to him with an assertion stating correctness. The voter can then choose to cast this vote, at which point the script submits his ballot and the assertion to the administrator. The administrator publishes the ballot and the assertion on the bulletin board. After some known deadline, the administrator homomorphically combines all ballots, and publishes the encrypted tally along with an assertion stating correctness of the tally. The trustees can then decrypt this tally, and the administrator publishes the result.

In Figure~\ref{fig:heliosassertions} we model Helios 2.0 with assertions. We do not include the final step, where the trustees decrypt the final encrypted tally and publish it onto the bulletin board. Note that this model, much like the terms-only model in~\cite{CS11}, requires us to add a homomorphic encryption operation to our term algebra. However, we can incorporate the weeding out of replayed ballots and establishing the link between ballots and voters by the use of assertions alone, instead of having to send extra terms. Note that in order for an agent $V_{2}$ to copy $V_{1}$'s vote and replay it to $A$, $V_{2}$ would need to make an assertion of the form $S\ \says\ \{ \exists v: b = \textit{ballot}(v) \conj V_{2}\ \says\ \texttt{valid}(v) \}$, which would contradict the sending in $V_{1}$'s name. Thus we can establish a link between vote and voter, while also disallowing replays. We merely need to add a homomorphic encryption operation to the term algebra, since our assertions, as of now, are not capable of capturing this operation.

\section{Assertions: Theory}
\label{sec:assertheory}

We fix the following countable sets -- a set $\vars$ of variables, a set $\ag$ of agents, a set $\nonces$ of nonces, and a set of $\keys$ of keys. We assume that every $k \in \keys$ has an inverse key, denoted $\inv(k)$. The set of basic terms $\basic$ is defined to be $\ag \cup \nonces \cup \keys$. The set of terms $\setofterms$ is given by the following syntax:
\[ t := m \pipe (t_{1}, t_{2}) \pipe \{t\}_{k}\]
where $m \in \basic \cup \vars$, and $k \in \keys \cup \vars$. A term with no variables occurring in it is called a \emph{ground term}. 

\begin{table}[h]
\centering
\small
\setlength{\tabcolsep}{0.5em}
\begin{tabular}{|Sc|Sc|}
\hline
\multicolumn{2}{|Sc|}{
\begin{math}
\begin{prooftree}
\justifies X \cup \{t\} \vdash t \using \rnax
\end{prooftree}
\end{math}
}
\\
\hline 
\begin{math}
\begin{prooftree}
X \vdash t_{1} \quad X \vdash t_{2} 
\justifies X \vdash (t_{1}, t_{2}) \using \rnpair
\end{prooftree}
\end{math}
&
\begin{math}
\begin{prooftree}
X \vdash (t_{1}, t_{2}) 
\justifies X \vdash t_{i} \using \rnsplit
\end{prooftree}
\end{math}
\\
\hline
\begin{math}
\begin{prooftree}
X \vdash t \quad X \vdash k
\justifies X \vdash \{t\}_{k} \using \rnenc
\end{prooftree}
\end{math}
&
\begin{math}
\begin{prooftree}
X \vdash \{t\}_{k} \quad X \vdash \inv(k)
\justifies X \vdash t \using \rndec
\end{prooftree}
\end{math}
\\
\hline
\end{tabular}
\caption{The Dolev-Yao derivation system}
\label{tab:dyterms}
\end{table} 
The system of rules for deriving new ground terms from old is given in Table~\ref{tab:dyterms}. The rules are presented in terms of sequents $X \vdash t$ where $X$ is a finite set of ground terms, and $t$ is a ground term.

\subsection{Assertions and derivations}
We now present the formal details of the model with assertions, a version of which was first proposed in \cite{RSS14}. 
The set of assertions, $\asserts$, is given by the following syntax (fixing a set of variables, and a set of predicates for each arity):
\begin{tabbing}
	$\alpha := t = t'$ \= $\pipe \alpha_1 \disj \alpha_2 \pipe \alpha_1 \conj \alpha_2 \pipe \exists{x}:\alpha \pipe m\ \says\ \alpha$ \\
	\> $\pipe \texttt{valid}(m) \pipe \texttt{elg}(m) \pipe \ldots \pipe m\ \sent\ t \pipe m\ \sent\ \alpha$
\end{tabbing}
where $t \in \setofterms$, $m \in \ag \cup \vars$, and \texttt{valid} and \texttt{elg} are application-specific predicates. The ellipses signify that one may add more such simple predicates, depending on the application requirements (as in the FOO protocol, from Section~\ref{subsec:fooasserts}).  A \emph{ground assertion} is one with no free variables.

The set of assertions is a positive fragment of existential first-order logic. The intention is that in addition to ground terms,  agents also communicate ground assertions to each other. Agents are allowed to assert equality of terms, and basic predicates on terms, as well as disjunctions and conjunctions. They can also ``sign'' assertions by use of the $\says$ operator. They also have the capability of existentially abstracting some terms from an assertion, thereby modelling \emph{witness hiding}. The sole use of the $\sent$ operator is to enable an observer to record who communicated a term or an assertion.

\begin{table}[ht]
{
\begin{center}
\setlength{\tabcolsep}{0.3em}
\begin{tabular}{|Sc|Sc|Sc|Sc|Sc|Sc|}
\hline
\multicolumn{2}{|Sc|}{
\begin{math}
\begin{prooftree}
	X \vdash_{\DY} m
	\justifies X, \Phi \vdash m = m \using [m \in \basic \cup \vars]
\end{prooftree}
\end{math}
}
&
\multicolumn{2}{Sc|}{
\begin{math}
\begin{prooftree}
	\justifies X, \Phi \cup \{\alpha\} \vdash \alpha \using \rnax
\end{prooftree}
\end{math}
}
&
\multicolumn{2}{Sc|}{
\begin{math}
\begin{prooftree}
	X, \Phi \vdash \alpha(t) \hspace{2mm} X, \Phi \vdash t = t'
	\justifies X, \Phi \vdash \alpha(t')
\end{prooftree}
\end{math}
}
\\
\hline
\multicolumn{2}{|Sc|}{
\begin{math}
\begin{prooftree}
	X, \Phi \vdash s = t \hspace{3mm} X, \Phi \vdash t = u
	\justifies X, \Phi \vdash s = u
\end{prooftree}
\end{math}
}
&
\multicolumn{2}{Sc|}{
\begin{math}
\begin{prooftree}
	X, \Phi \vdash s = t
	\justifies X, \Phi \vdash t = s
\end{prooftree}
\end{math}
}
&
\multicolumn{2}{Sc|}{
\begin{math}
\begin{prooftree}
	X, \Phi \vdash (s_{0}, s_{1}) = (t_{0}, t_{1})
	\justifies X, \Phi \vdash s_{i} = t_{i}
\end{prooftree}
\end{math}
}
\\
\hline
\multicolumn{2}{|Sc|}{
\begin{math}
\begin{prooftree}
	X, \Phi \vdash s = s' \hspace{2mm} X, \Phi \vdash t = t' 
	\justifies X, \Phi \vdash (s,t) = (s',t')
\end{prooftree}
\end{math}
}
&
\multicolumn{2}{Sc|}{
\begin{math}
\begin{prooftree}
	X, \Phi \vdash \{s_{0}\}_{s_{1}} = \{t_{0}\}_{t_{1}}
	\justifies X, \Phi \vdash s_{i} = t_{i} \using \bullet
\end{prooftree}
\end{math}
}
&
\multicolumn{2}{Sc|}{ 
\begin{math}
\begin{prooftree}
	X, \Phi \vdash s = s' \hspace{2mm} X, \Phi \vdash m = m' 
	\justifies X, \Phi \vdash \{s\}_{m} = \{s'\}_{m'}
\end{prooftree}
\end{math}
}
\\
\hline
\multicolumn{4}{|Sc|}{
\begin{math}
\begin{prooftree}
	X, \Phi \vdash m = n
	\justifies X, \Phi \vdash \alpha \using \rncontr\ [m, n \in \basic, m \neq n]
\end{prooftree}
\end{math}
}
&
\multicolumn{2}{Sc|}{
\begin{math}
\begin{prooftree}
	X, \Phi \vdash \alpha \hspace{3mm} X \vdash_{\DY} \priv(A)
	\justifies X, \Phi \vdash A \ \says\  \alpha \using \rnsays{A}
\end{prooftree}
\end{math}
}
\\
\hline
\multicolumn{2}{|Sc|}{
\begin{math}
\begin{prooftree}
	X, \Phi \vdash \alpha \hspace{2mm} X, \Phi \vdash \beta
	\justifies X, \Phi \vdash \alpha \conj \beta \using \rnandintro{}
\end{prooftree}
\end{math}
}
&
\multicolumn{2}{Sc|}{
\begin{math}
\begin{prooftree}
	X, \Phi \vdash \alpha_{1} \conj \alpha_{2}
	\justifies X, \Phi \vdash \alpha_{i} \using \rnandelim{}
\end{prooftree}
\end{math}
}
&
\multicolumn{2}{Sc|}{
\begin{math}
\begin{prooftree}
	X, \Phi \vdash A\ \says\ \alpha
    \justifies X, \Phi \vdash \alpha \using \textit{strip}
\end{prooftree}
\end{math}
}
\\
\hline
\multicolumn{2}{|Sc|}{
\begin{math}
\begin{prooftree}
X, \Phi \vdash \alpha_i
\justifies X, \Phi \vdash \alpha_1 \disj \alpha_2 \using \rnorintro{}
\end{prooftree}
\end{math}
}
&
\multicolumn{4}{Sc|}{
\begin{math}
\begin{prooftree}
	X, \Phi \vdash \alpha \disj \beta \hspace{3mm} X, \Phi\cup\{\alpha\} \vdash \delta \hspace{3mm} X, \Phi \cup \{\beta\} \vdash \delta
	\justifies X, \Phi \vdash \delta \using \rnorelim{}
\end{prooftree}
\end{math}
}
\\
\hline
\multicolumn{2}{|Sc|}{
\begin{math}
\begin{prooftree}
	X, \Phi \vdash \alpha(t)
    \justifies X, \Phi \vdash \exists x: \alpha(x) \using \rnexistsintro
\end{prooftree}
\end{math}
}
&
\multicolumn{4}{Sc|}{
\begin{math}
\begin{prooftree}
	X, \Phi \vdash \exists{}x:\alpha(x) \hspace{3mm} X, \Phi \cup \{\alpha(y)\} \vdash \beta
	\justifies X, \Phi \vdash \beta \using \rnexistselim
\end{prooftree}
\end{math}
}
\\
\hline
\end{tabular}
\end{center}
}
\caption{Derivation rules for assertions. We assume that $X \DYderives x$ for all variables $x$, and that $\inv(x) = x$. In the $\bullet$ rule, we require that $X \DYderives \inv(s_{1})$ and $X \DYderives \inv(t_{1})$. In the $\rnexistselim$ rule, we require that $y \not\in \textit{Vars}(X, \Phi \cup \{\beta\})$.}
\label{tab:assrules}
\end{table}

%

In the course of participating in a protocol, agents accumulate a database of ground terms and ground assertions communicated to them. The proof system for assertions is presented in Table~\ref{tab:assrules}. The rules are presented in terms of sequents $X, \Phi \vdash \alpha$, where $X$ is a finite set of ground terms and $\Phi$ is a finite set of assertions (which are not necessarily ground). 

Equality assertions form a central part of communications between agents. Note that an agent $A$ can derive $t=t$ only when all basic subterms of $t$ can be derived by $A$. The recipient of an equality assertion can use the rules provided in Table~\ref{tab:assrules} to reason further about the terms involved therein. Our rules for equality are fairly intuitive and reflect basic properties of the pairing and encryption operations. Equality assertions are most likely to be used in existentially quantified assertions. Notable among the other rules are $\rnsays{A}$, which allows the possessor of $\priv(A)$ to ``sign'' an assertion in $A$'s name, and $\textit{strip}$, which allows one to strip the sign in $A\ \says\ \alpha$ and use $\alpha$ in local reasoning.

These rules allow agents to carry out non-trivial inferences, potentially learning more than was intended by the protocol. Suppose an agent $A$ has a term $\{v\}_{k}$, which he knows be a nonce encrypted with some key, but whose inverse he does not have access to. One would presume that $A$ therefore should have no idea about the value of $v$. However, it is possible for assertions about $\{v\}_{k}$ to reveal more information to $A$. Suppose $A$ manages to obtain two certificates $\exists{}x,y: \{v\}_{k} = \{x\}_{y}\ \conj\ (x = 0 \disj x = 1)$ and $\exists{}x,y: \{v\}_{k} = \{x\}_{y} \conj (x = 0 \disj x = 2)$. Let us call these assertions $\exists{}x,y: \alpha(x,y)$ and $\exists{}x,y: \alpha'(x,y)$. These two assertions are in $A$'s database of assertions $\Phi$. Let $a,b,a',b'$ be new variables that do not occur in $\Phi$. Consider $\Phi \cup \{\alpha(a,b), \alpha'(a',b')\}$. From $\{v\}_{k} = \{a\}_{b}$ and $\{v\}_{k} = \{a'\}_{b'}$, we get $\{a\}_{b} = \{a'\}_{b'}$, and hence $a = a'$ and $b = b'$. From the other parts of $\alpha$ and $\alpha'$, and using transitivity, we get $a = 0 \disj a = 1$ and $a = 0 \disj a = 2$. We use disjunction elimination to get $a = 0$. From this we conclude that $\{v\}_{k} = \{0\}_{b}$, and hence $\Phi \cup \{\alpha(a,b), \alpha'(a',b')\} \vdash \exists{}y:(\{v\}_{k} = \{0\}_{y})$. Therefore, using the $\rnexistselim$ rule, we get $\Phi \vdash \exists{}y:(\{v\}_{k} = \{0\}_{y})$.

In the formal model of \cite{BHM08,BMU08}, each zkp term proves a formula involving some private and some public variables. The recipient of a zkp term is deemed to have knowledge of the terms used in place of the public variables, but not the private ones. We adopt a similar convention. For an assertion $\alpha$, if an equality of the form $t=t'$ occurs in it, or if $\alpha$ involves the application of a predicate to a term $t$, then $\alpha$ reveals $t$. However, if a term of the form $\{v\}_{k}$, say, appears in $\alpha$, then $\alpha$ does not reveal $v$. We also adopt the convention that every term revealed by an assertion is sent earlier in the protocol.

\subsection{Actions, roles and protocols}
\label{subsec:roles}
There are six type of actions -- send, anonymous send, receive, confirm, deny, and insert. Sends, anonymous sends, and receives are of the form $\sendact{A}{\vec{m}}{t}{\alpha}$, $\sendact{A^{*}}{\vec{m}}{t}{\alpha}$ and $\recact{A}{t}{\alpha}$ respectively, where $A \in \ag \cup \{\agentid\}$ (where $\agentid$ is a dedicated variable that stands for the agent performing the action), $\vec{m} \subseteq \vars \cup \nonces \cup \keys$ stands for nonces and keys that are \emph{fresh} which should be instantiated with hitherto unused values in each occurrence of this action, $t \in \setofterms$ and $\alpha \in \asserts$. The $\confirmact{A}{\alpha}$ and $\denyact{A}{\alpha}$ actions allow $A$ to branch on whether or not he can derive $\alpha$, while $\insertact{A}{\alpha}$ allows $A$ to add previously unknown true assertions into her database. For $A \in \ag \cup \{\agentid\}$, an $A$-action is an action which involves $A$. A \emph{ground action} is one without any variable occurrence. An $A$-role is a finite sequence of $A$-actions. A role is an $A$-role for some $A \in \ag \cup \{\agentid\}$. A \textit{protocol} $\prot$ is a finite set of roles.

Given a sequence of actions $\eta = a_{1}\cdots{}a_{n}$, we say that the variable $x$ \emph{originates at} $i$ if $x$ occurs in $a_{i}$ and does not occur in $a_{j}$ for any $j < i$. A variable $x$ occurring in a role $\eta$ is said to be \emph{bound} if it originates at $i$ and either $a_{i}$ is a receive action, or $a_{i} = \sendact{A}{\vec{y}}{t}{\alpha}$ is a send action with $x \in \vec{y}$.

As an example, we show the voter role for the FOO protocol from Section~\ref{sec:foo}. In this role, $v$ and $\agentid$ stand for the vote and voter respectively, while $k, k'$ are fresh keys, and $\textit{auth}$ is a bound variable (since it originates in a receive) which stands for the authority with whom the voter interacts.

\begin{eqnarray*}
			+\agentid &:& (k)\ \{v\}_{k}, \agentid\ \says\ \{\exists{}x,r:\{x\}_{r} = \hidden{\{v\}_{k}} \conj \texttt{valid}(x)\} \\
			-\agentid &:& \textit{auth}\ \says\ [\texttt{elg}(\agentid) \ \conj\ \texttt{voted}(\agentid, \{v\}_{k}) \\
			& & \ \conj\ \agentid\ \says\ \{\ \exists x, r: \{x\}_{r} = \{v\}_{k} \conj\ \texttt{valid}(x)\ \}\ ] \\
			+\agentid^{*}&:&(k')\ (\{v\}_{k'}, k'), \\
			& & \exists{}X,y,s: \textit{auth}\ \says\ [ \texttt{elg}(X) \conj \texttt{voted}(X, \{y\}_{s}) \\
			& & \conj\ X\ \says\ \{ \exists{}x,r: \{x\}_{r} = \{y\}_{s} \conj \texttt{valid}(x) \} ] \conj\ y = v
\end{eqnarray*}

The authority and counter roles can also be extracted from the protocol description in a similar manner.

\subsection{Runs of a protocol}
Even though the roles of a protocol mention variables, its \textit{runs} (or executions) consist only of ground terms and assertions exchanged in various \emph{instances} of the roles. An instance of a role is formally specified by a \emph{substitution} $\sigma$, which is a partial map from $\vars$ to the set of all ground terms. We lift $\sigma$ for terms, assertions and actions in the standard manner. $\sigma$ is said to be \emph{suitable} for an action $a$ if $\sigma(a)$ is an action, i.e. a typing discipline is followed. A substitution is suitable for a role $\eta$ if it is defined on all free variables of $\eta$ and suitable for all actions in $\eta$. 

A \textit{session} of a protocol $\prot$ is a sequence of actions of the form $\sigma(\eta)$, where $\eta \in \prot$ and $\sigma$ is suitable for $\eta$. 

A run of a protocol is an interleaving of sessions in which each agent can construct the messages that it communicates. This is formalized by a notion of \emph{knowledge state}, which represents all the terms and assertions that each agent knows. A \emph{control state} is a record of progress made by an agent in the various sessions he/she participates in. 


A knowledge state $\ks$ is a tuple $((X_A, \Phi_A)_{A \in \ag})$, where $X_A$ (resp. $\Phi_A$) is the set of ground terms (resp. ground assertions) belonging to an agent $A$. A control state $S$ is a finite set of sequences of actions. A protocol state is a pair $(\ks, S)$ where $\ks$ is a knowledge state and $S$ is a control state. 

%

\begin{definition}
Let $(\ks, S)$ and $(\ks', S')$ be two states of a protocol $\prot$, and let $b$ be a ground action. We say that $(\ks, S) \xrightarrow{b} (\ks', S')$ iff there is a session $\eta = a\cdot\eta' \in S$ and a substitution $\sigma$ suitable for $\eta'$ such that:
\begin{itemize}
	\item $b = \sigma(a)$
	\item $S' = (S \setminus \{\eta\}) \cup \{\sigma(\eta')\}$
	\item $\ks \xrightarrow{b} \ks'$ as given in Table~\ref{tab:updates}. 
\end{itemize}
\label{def:updates}
\end{definition}

In Definition~\ref{def:updates}, we add $\sigma(\eta')$ rather than $\eta'$, in order to update the substitution associated with the session on executing the action. This update reflects the new values generated for each fresh nonce variable (in case the action is a send) or the new bindings for input variables (in case the action is a receive). For instance, if $\eta = a\cdot\eta'$ where $a = \recact{A}{(x,y)}{\alpha(x,y)}$ and $b = \recact{A}{(t,t')}{\alpha(t,t')}$, then $\sigma = [x:=t, y:=t']$. Any occurrence of $x$ in $\eta'$ is bound to $t$. 


\begin{table}[h]
\centering
\setlength{\tabcolsep}{0.4em}
\renewcommand{\arraystretch}{1.2}
\begin{tabular}{|c|c|c|}
\hline
Action $b$ & Enabling conditions & Updates\\
\hline
\multirow{2}{*}{$\sendact{A}{\vec{m}}{t}{\alpha}$} & $X_{A} \cup \vec{m} \vdash t$ & $X'_{A} = X_{A} \cup \vec{m}$ \ \ $X'_{I} = X_{I} \cup \{t\}$ \\ 
& $X_{A} \cup \vec{m}, \Phi_{A} \vdash \alpha$ & $\Phi'_{I} = \Phi_{I} \cup \{\alpha, A\ \sent\ t, A\ \sent\ \alpha\}$ \\
\hline
\multirow{2}{*}{$\sendact{A^{*}}{\vec{m}}{t}{\alpha}$} & $X_{A} \cup \vec{m} \vdash t$ & $X'_{A} = X_{A} \cup \vec{m}$ \ \ \ $X'_{I} = X_{I} \cup \{t\}$ \\
& $X_{A} \cup \vec{m}, \Phi_{A} \vdash \alpha$ & $\Phi'_{I} = \Phi_{I} \cup \{\alpha\}$  \\
\hline
\multirow{2}{*}{$\recact{A}{t}{\alpha}$} & $X_{I} \vdash t$ & $X'_{A} = X_{A} \cup \{t\}$ \\
& $X_{I}, \Phi_{I} \vdash \alpha$ & $\Phi'_{A} = \Phi_{A} \cup \{\alpha\}$ \\
\hline
$\confirmact{A}{\alpha}$ & $X_{A}, \Phi_{A} \vdash \alpha$ & No change \\
\hline
$\denyact{A}{\alpha}$ & $X_{A}, \Phi_{A} \nvdash \alpha$ & No change \\
\hline
$\insertact{A}{\alpha}$ & Always enabled & $\Phi'_{A} = \Phi_{A} \cup \{\alpha\}$ \\ 
\hline
\end{tabular}
\caption{Enabling conditions for $\ks \xrightarrow{b} \ks'$. We assume that for each agent $A$, $(X_{A}, \Phi_{A})$ and $(X'_{A}, \Phi'_{A})$ represent $A$'s knowledge in $\ks$ and $\ks'$, respectively.}
\label{tab:updates}
\end{table}

Note the crucial difference between the updates for sends and anonymous sends -- in the former, the intruder updates its state with $A\ \sent\ t$ and $A\ \sent\ \alpha$, whereas in the latter, no sender information is available to any observer (including the intruder).

An \emph{initial control state} of $\prot$ is a finite set of sessions of $\prot$. In the \emph{initial knowledge state}, each agent has her own secret keys and shared keys, all public keys in her database, and potentially some constants of $\prot$.
 
\begin{definition}
A \textit{run} of a protocol $\prot$ is $(\ks_{0}, a_{1}\cdots a_{n})$ such that $\ks_{0}$ is an initial knowledge state, and there exist sequences $\ks_{1}, \ldots, \ks_{n}$ and $S_{0}, \ldots, S_{n}$ such that $(\ks_{i-1}, S_{i-1}) \xrightarrow{a_{i}} (\ks_{i}, S_{i})$ for all $i \leq n$.
\end{definition}

\subsection{Notes on implementability}
A central aspect of this model is that communicated assertions are ``believed'' by the recipients. This is reflected in the updates for receive actions. On the other hand, it is not possible for a malicious agent to inject ``falsehoods'' into the system, as evidenced by the enabling conditions which only allow derivable assertions to be communicated. How might all this be realized in practice? 

An implementation is to demand that every communicated assertion be translated into an appropriate zero knowledge proof. But suppose an agent receives ZKPs for assertions $\alpha$ and $\beta$ from $A$ and $B$, and wishes to send $\alpha \conj \beta$ to someone else. For this, she should have the capacity to produce a ZKP for $\alpha \conj \beta$. This implements the $\rnandintro$ rule in our system. Clearly this requires some mechanism for composing ZKPs. Such a system has been studied in~\cite{MPR13}, which proposes a logical language close to ours, and also discusses modular construction of ZKPs, based on the seminal work on composability of ZKPs~\cite{GS08}.

However, \cite{MPR13} has some restrictions on the proof rules for which one can modularly construct ZKPs. For instance, they do not consider disjunction elimination or existential elimination. Nevertheless, we consider these rules since they are at the heart of potential attacks (as illustrated by the earlier example). This situation can be handled formally by making a distinction between rules that are ``safe for composition'' and rules that are not. A rule like $\rnandintro$ is safe for composition, for example, whereas $\rnorelim$ might not be. We then adopt the restriction that we communicate assertions that are derived using only safe rules. If the derivation of an assertion necessarily involves unsafe rules, then it cannot be communicated to another agent, even though this derivation itself is allowed for local reasoning. In this paper, we therefore consider both local reasoning to derive more assertions (to gain more knowledge about some secrets, for instance) as well as deriving communicable assertions.

\section{Formalizing anonymity}
\label{sec:unlink}

Informally, we say that a voting protocol satisfies anonymity if in all executions of the protocol, no adversary can deduce the connection between a voter and her vote. One way to formalise it is to consider a run $\rho$ where voter $V_{0}$ voted $0$ and voter $V_{1}$ voted $1$, and show that there is some run $\rho'$ where the votes of $V_{0}$ and $V_{1}$ are swapped and every other voter acts the same as in $\rho$, such that even the most powerful intruder $I$ (who has access to all keys of the authorities) cannot distinguish $\rho$ from $\rho'$. 


\begin{definition}
Let $(\ks, \rho)$ and $(\ks', \rho')$ be two runs of $\prot$, where $\rho = a_{1}\cdots{}a_{n}$ and $\rho' = a'_{1}\cdots{}a'_{n}$. Let $t_{i}$ and $t'_{i}$ be the terms communicated in $a_{i}$ and $a'_{i}$, respectively. Let $(X, \Phi)$ and $(X', \Phi')$ be the knowledge states of $I$ at the end of each run.
	
We say that $(\ks, \rho)$ is \textit{$I$-indistinguishable} from  $(\ks', \rho')$ -- denoted $(\ks, \rho) \sim_{I} (\ks', \rho')$ -- if for all assertions $\alpha(x_{1}, \ldots, x_{k})$ and all sequences $i_{1} < \cdots < i_{k} \leq n$: 
	\[
	X, \Phi \derives \alpha(t_{i_{1}}, \ldots, t_{i_{k}})\ \text{iff}\ X', \Phi' \derives \alpha(t'_{i_{1}}, \ldots, t'_{i_{k}}).
	\]
\end{definition}

One can view the parameters $x_{1}, \ldots, x_{k}$ occurring in the above definition as \emph{handles}, and the mapping from $x_{1}, \ldots, x_{k}$ to $t_{i_{1}}, \ldots, t_{i_{k}}$ as an \emph{active substitution}. Parametrized assertions $\alpha(x_{1}, \ldots, x_{k})$ constitute \emph{tests} on each run of the protocol. Thus the above notion is related to the notion of static equivalence that is central to protocol modelling in the applied-pi calculus~\cite{BHM08,BMU08,KR05}. Note that the notion of indistinguishability we use here is trace-based, as that fits naturally with our model. But it is also possible to have a bisimulation-based definition, and adapt our proof ideas. 

Consider a voting protocol $\prot$ with three roles -- \emph{voter}, \emph{authority} and \emph{counter}, and two phases: \emph{authorization} and \emph{voting}. For simplicity, we assume that there are two fixed agents $A$ and $C$ who play the authority and counter role, respectively. If there is only one voter in a run, then obviously his/her vote can be linked to him/her. If a voter's vote is counted during the authorization phase, then we might have a situation where a vote is cast by a voter before anyone else has been authorized. This again is an easy violation of anonymity. Therefore we assume that in any run of $\prot$, there are at least two agents playing the voter role, and all $V_{i} \rightarrow A$ actions precede all $V_{j} \rightarrow C$ actions. 
	
Fix voter names $V_{0}$, $V_{1}$, and votes $v_{0}$ and $v_{1}$. A session $\eta$ of $\prot$ is said to be an $(i,j)$-session if $\eta$ maps $\agentid$ to $V_{i}$ and $v$ to $v_{j}$. 
\begin{definition}	
	We say that $\prot$ \textit{satisfies anonymity} if for every initial knowledge state $\ks = (X, \Phi)$ such that $X_{A} \cup X_{C} \subseteq X_{I}$, and for every run $(\ks, \rho)$ which includes a $(0,0)$-session and a $(1,1)$-session, there is a run $(\ks, \rho')$ which includes a $(1,0)$-session and a $(0,1)$-session such that $(\ks, \rho) \sim_{I} (\ks, \rho')$.
\end{definition}

\begin{theorem}
	The FOO protocol satisfies anonymity.
\end{theorem}
\begin{proof}
Recall the voter role for FOO from Section~\ref{subsec:roles}. Consider a run $(\ks, \rho)$ of FOO whose initial control state is $S \cup \{\eta_{0}, \eta_{1}\}$, where $\eta_{0}$ is the $(0,0)$-session and $\eta_{1}$ is the $(1,1)$-session. Let $\eta_{2}$ and $\eta_{3}$ be the $(0,1)$-session and $(1,0)$-session, respectively. We construct a run $\rho'$ which includes $\eta_{2}$ and $\eta_{3}$ such that $(\ks, \rho) \sim_{I} (\ks, \rho')$. The session $\eta_{0}$ assigns values $p$ and $r$ to the keys $k$ and $k'$ from the role description, while $\eta_{1}$ assigns values $q$ and $s$ respectively. For ease of notation, we denote $v_{0}$ and $v_{1}$ by $u$ and $v$ respectively, and $d = \{u\}_{p}$ and $e = \{v\}_{q}$. 

Suppose $\rho = a_{1}\cdots{}a_{n}$. Assume without loss of generality that both sessions $\eta_{0}$ and $\eta_{1}$ are fully played out in $\rho$. Also without loss of generality, let $i < j < k < l$ be indices such that the send actions of $\eta_{0}$ are $a_{i}$ and $a_{k}$, and the send actions of $\eta_{1}$ are $a_{j}$ and $a_{l}$, where
\begin{eqnarray*}
a_{i} = \sendact{V_{0}}{p}{d}{\beta(d)} &\text{and}&  
a_{k} = \sendact{V_{0}^{*}}{r}{(\{u\}_{r}, r)}{\gamma(u)}	\\
a_{j} = \sendact{V_{1}}{q}{e}{\beta(e)} &\text{and}&  
a_{l} = \sendact{V_{1}^{*}}{s}{(\{v\}_{s}, s)}{\gamma(v)}	
\end{eqnarray*}
We build $\rho' = b_{1}\cdots{}b_{n}$ as shown in Figure~\ref{fig:runswap}. 
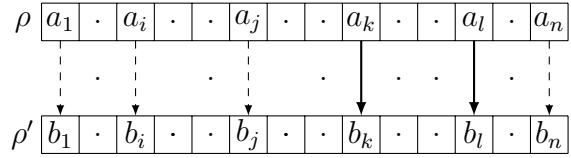
\begin{figure}[h]
\begin{center}
\begin{tikzpicture}
\node             at (-0.25,-1.25)   (A) {$\rho$}; 
\draw[step=0.5cm] (0,-1.5) grid (7,-1); 

\node			  at (0.25,-1.25)      (S1) {$a_{1}$};
\node			  at (0.75,-1.25)      {$.$};
\node			  at (1.25,-1.25)      (V11) {$a_{i}$};
\node			  at (1.75,-1.25)      {$.$};
\node			  at (2.25,-1.25)      {$.$};
\node			  at (2.75,-1.25)      (V12) {$a_{j}$};
\node			  at (3.25,-1.25)      {$.$};
\node			  at (3.75,-1.25)      {$.$};
\node			  at (4.25,-1.25)      (V13) {$a_{k}$};
\node			  at (4.75,-1.25)      {$.$};
\node			  at (5.25,-1.25)      {$.$};
\node			  at (5.75,-1.25)      (V14) {$a_{l}$};
\node			  at (6.25,-1.25)      {$.$};
\node			  at (6.75,-1.25)	   (S1n) {$a_{n}$};

\node             at (-0.25,-2.75)   (B) {$\rho'$}; 

\draw[step=0.5cm] (0,-3) grid (7,-2.5); 

\node			  at (0.25,-2.75)      (S2) {$b_{1}$};
\node			  at (0.75,-2.75)      {$.$};
\node			  at (1.25,-2.75)      (V21) {$b_{i}$};
\node			  at (1.75,-2.75)      {$.$};
\node			  at (2.25,-2.75)      {$.$};
\node			  at (2.75,-2.75)      (V22) {$b_{j}$};
\node			  at (3.25,-2.75)      {$.$};
\node			  at (3.75,-2.75)      {$.$};
\node			  at (4.25,-2.75)      (V23) {$b_{k}$};
\node			  at (4.75,-2.75)      {$.$};
\node			  at (5.25,-2.75)      {$.$};
\node			  at (5.75,-2.75)      (V24) {$b_{l}$};
\node			  at (6.25,-2.75)      {$.$};
\node			  at (6.75,-2.75)	   (S2n) {$b_{n}$};

\node			  at (0.75,-2)      {$.$};
\node			  at (1.75,-2.75)   {$.$};
\node			  at (2.25,-2)      {$.$};
\node			  at (3.25,-2.75)   {$.$};
\node			  at (3.75,-2)      {$.$};
\node			  at (4.75,-2)      {$.$};
\node			  at (5.25,-2)      {$.$};
\node			  at (6.25,-2)      {$.$};

\path (4.25,-1.5) edge [->,thick,>=latex] (4.25,-2.5); 
\path (5.75,-1.5) edge [->,thick,>=latex] (5.75,-2.5); 
\path (0.25,-1.5) edge [dashed,->,>=latex] (0.25,-2.5); 
\path (1.25,-1.5) edge [dashed,->,>=latex] (1.25,-2.5); 
\path (2.75,-1.5) edge [dashed,->,>=latex] (2.75,-2.5); 
\path (6.75,-1.5) edge [dashed,->,>=latex] (6.75,-2.5); 

\end{tikzpicture}
\end{center}
\caption{Building $\rho'$ from $\rho$. The dashed arrows capture $b_{m} = a_{m}[d \mapsto e, e \mapsto d]$, for all $m \not\in \{l,k\}$. For $m \in \{l, k\}$, the thick arrows stand for $b_{m} = a_{m}[V_{0} \mapsto V_{1}, V_{1} \mapsto V_{0}]$.}
\label{fig:runswap}
\end{figure}

Observe that $\rho'$ is also a run of FOO starting from the state $(\ks, S \cup \{\eta_{2}, \eta_{3}\})$, where $\eta_{2}$ contains $b_{i}$ and $b_{l}$, and $\eta_{3}$ contains $b_{j}$ and $b_{k}$. We crucially use the fact that we do not fix the instances of the fresh nonces a priori, so we can swap the action containing $p$ as a fresh nonce with the one containing $q$ as a fresh nonce, for example.

	

For any term $t$ (resp. assertion $\alpha$), we define $\swp(t)$ (resp. $\swp(\alpha)$) to be the result of changing all occurrences of $d$ to $e$ and vice versa. $\swp$ is lifted to sets of terms and assertions as usual.

Let $(X, \Phi)$ and $(X', \Phi')$ be the knowledge states of $I$ at the end of $\rho$ and $\rho'$ respectively. It is evident from the construction of $\rho'$ that $X' = \swp(X)$. Furthermore, it is easy to see that neither $X$ nor $X'$ derive either $p$ or $q$, and that $X \DYderives t$ iff $X' \DYderives \swp(t)$.   

It can also be seen that $\Phi' = \swp(\Phi)$, as elaborated below.
For every $m$, if $a_{m}$ communicates $\alpha$, then $b_{m}$ communicates $\swp(\alpha)$. The other formulas added to $\Phi$ are $\sent$ assertions. For every action $a_{m}$ other than $a_{k}$ and $a_{l}$, the sender of $b_{m}$ is unchanged from $a_{m}$. Therefore, a $\sent$ assertion with the same sender name would be added to $\Phi$ and $\Phi'$. For $a_{k}$ and $a_{l}$, no $\sent$ assertions are added since these are anonymous sends. Therefore, $\Phi' = \swp(\Phi)$. 

We now prove that $X, \Phi \derives \alpha(t_{i_{1}}, \ldots, t_{i_{k}})$ iff $X', \Phi' \derives \alpha(t'_{i_{1}}, \ldots, t'_{i_{k}})$, for all assertions $\alpha(x_{1}, \ldots, x_{k})$. It suffices to prove that $X, \Phi \derives \alpha$ iff $X', \Phi' \derives \swp(\alpha)$ for all $\alpha$. For every $\exists{}:\delta$, let $y_{\delta}$ be a variable that does not occur in $\Phi$. A set $\Theta$ is said to be \emph{closed under witnesses} if $\delta(y_{\delta}) \in \Theta$ for all $\exists{}y:\delta \in \Theta$. Let $\Pi$ be the smallest superset of $\Phi$ closed under witnesses. We use $\Pi'$ to denote $\swp(\Pi)$. It can be shown by an analysis of derivations that $X, \Phi \vdash \alpha$ iff $X, \Pi \vdash_{1} \alpha$ and $X', \Phi' \vdash \alpha$ iff $X', \Pi' \vdash_{1} \alpha$, where $\vdash_{1}$ denotes derivability without using the $\rnexistselim$ rule. Note that both $X, \Pi$ and $X', \Pi'$ are safe for $d$ and $e$ in the following sense. They do not derive equalities of the form $p = t$ or $q = t$ for any term $t$, and they do not derive equalities of the form $d = t'$ or $e = t'$ where $t'$ is a term containing a non-variable. We now prove the final claim needed for indistinguishability of $\rho$ and $\rho'$.

\medskip
\noindent \textbf{Claim.}\qquad For any $\alpha$, $X, \Pi \vdash_{1} \alpha$ iff $X', \Pi' \vdash_{1} \swp(\alpha)$.

\noindent \textbf{Proof of Claim}\qquad 
We prove the implication from left to right, by induction on structure of derivations. The other direction holds by symmetry. Suppose $\pi$ is a derivation of $X, \Pi \derives \alpha$, with last rule $r$. 
 \begin{description}
 	\item[$r = \rnax$:] Suppose $\alpha \in \Pi$. It follows that $\swp(\alpha) \in \Pi'$.
 	\item[$r$ is equality of encrypted terms:] $\pi$ looks as follows.
 	\[
 		\begin{prooftree}
 			\[\pi_{0} \leadsto \justifies X, \Pi \vdash s = s'\]
 			\[\pi_{1} \leadsto \justifies X, \Pi \vdash m = m'\]
 			\justifies X, \Pi \vdash \{s\}_{m} = \{s'\}_{m'}
 		\end{prooftree}
 	\]
 	Suppose $\{s\}_{m}$ is either $d$ or $e$. Then $m$ is either $p$ or $q$, and this would mean that $p = m'$ or $q = m'$ is derivable, contradicting safety of $X, \Pi$. Therefore $\{s\}_{m}$ is not equal to either $d$ or $e$. By induction hypothesis, $\swp(s = s')$ is derivable from $X', \Pi'$, and hence $\swp(\{s\}_{m} = \{s'\}_{m'})$ is also derivable. 
 	\item[$r$ is equality of decrypted terms:] In this case, $\pi$ is of the following form
 	\[
 		\begin{prooftree}
 			\[\pi_{0} \leadsto \justifies X, \Pi \vdash \{s\}_{m} = \{s'\}_{m'}\]
			\[\pi_{1} \leadsto \justifies X \DYderives \inv(m)\]
			\[\pi_{2} \leadsto \justifies X \DYderives \inv(m')\]
 			\justifies X, \Pi \vdash s = s'
 		\end{prooftree}
 	\]
	By induction hypothesis, it follows that $X', \Pi' \vdash \swp(\{s\}_{m}) = \swp(\{s'\}_{m'})$. Observe that neither $\{s\}_{m}$ nor $\{s'\}_{m'}$ is the same as $d$ or $e$ (for otherwise we would have that $X \DYderives p$ or $X \DYderives q$, which is an impossibility). Thus any occurrence of $d$ or $e$ in $\{s\}_{m}$ is inside $s$, and similarly for $\{s'\}_{m'}$. Thus $\swp(\{s\}_{m}) = \{\swp(s)\}_{m}$ and $\swp(\{s'\}_{m'}) = \{\swp(s')\}_{m'}$. Therefore $\swp(s) = \swp(s')$ is also derivable. ($\inv(m)$ and $\inv(m')$ are derivable from $X'$ since they are derivable from $X$ and do not mention $d$ or $e$.) 
	%
	\end{description}
The rest of the cases are on similar lines (or simpler, appealing to the induction hypothesis).
\end{proof}

\section{Conclusion}
In this paper, we extended the model of~\cite{RSS14} by adding exisential assertions to the language, as a tool to hide private data used to generate certificates. These assertions are especially useful in coding up constructs that are common to voting protocols. We showed how to specify protocols in this model, and formalised the notion of anonymity in terms of indistinguishability. In a non-trivial example of analysis in our model, we proved anonymity for the FOO protocol. 

One way of extending this model is by adding a background theory of universally quantified sentences. Such a theory is a standard part of many authorization systems. For instance, if an agent $A$ communicates to $B$ the assertion $\exists{x}:\texttt{voted}(V,x)$ and if the background theory contains the assertion 
\[ \forall{X},{x}: \{ \texttt{voted}(X,x) \Rightarrow \texttt{elg}(X) \} \]
then $B$ can conclude $\texttt{elg}(V)$. More detailed examples are found in~\cite{BHM08,MPR13}. It is an important ingredient in many systems, and we can easily incorporate it in our theoretical model.

\end{document}